\documentclass[intlimits]{birkmult}
\usepackage{amsmath}
\usepackage{float}
\usepackage{cite}
\usepackage{upref}
\usepackage{url}

\numberwithin{equation}{section}

\theoremstyle{plain}
\newtheorem{theorem}{Theorem}[section]
\newtheorem{corollary}{Corollary}[section]

\theoremstyle{remark}
\newtheorem{remark}{Remark}[section]

\newcommand{\Area}{\operatorname{Area}}
\newcommand{\clos}{\operatorname{clos}}
\newcommand{\const}{\operatorname{const}}
\newcommand{\grad}{\operatorname{grad}}
\newcommand{\mbb}{\mathbb}
\newcommand{\supp}{\operatorname{supp}}
\newcommand{\tu}{\textup}

\renewcommand{\[}{\left[ }
\renewcommand{\]}{\right] }
\renewcommand{\(}{\left( }
\renewcommand{\)}{\right) }

\title{Gravitational Lensing by Elliptical Galaxies, and the Schwarz Function}
\author[Fassnacht]{C. D. Fassnacht}
\address{Department of Physics \\
University of California at Davis \\
Davis, CA 95616}
\email{fassnacht@physics.ucdavis.edu}
\author[Keeton]{C. R.  Keeton}
\address{Department of Physics and Astronomy\\
Rutgers University \\
Piscataway, NJ 08854-8019} \email{keeton@physics.rutgers.edu}
\author[Khavinson]{D. Khavinson}
\address{Department of Mathematics \& Statistics \\
University of South Florida \\
Tampa, FL 33620-5700}
\email{dkhavins@cas.usf.edu}
\thanks{The third author gratefully acknowledges partial support from the National Science Foundation under the grant DMS-0701873. The first and third
authors are also grateful to Kavli Institute of Theoretical Physics
for the partial support of their visit there in 10/2006 under the
NSF grant PHY05-51164.}
\date{}

\begin{document}
\maketitle

\begin{abstract}
We discuss gravitational lensing by elliptical galaxies with some
particular mass distributions. Using simple techniques from the
theory of quadrature domains and the Schwarz function (cf.\
\cite{Sh}) we show that when the mass density is constant on
confocal ellipses, the total number of lensed images of a point
source cannot exceed $5$ ($4$ bright images and $1$ dim image).
Also, using the Dive--Nikliborc converse of the celebrated Newton's
theorem concerning the potentials of ellipsoids, we show that
``Einstein rings'' must always be either circles (in the absence of
a tidal shear), or ellipses.
\end{abstract}

\section{Basics of gravitational lensing}

Imagine $n$ co-planar point-masses (e.g., condensed galaxies, stars,
black holes) that lie in one plane, the lens plane. Consider a point
light source $S$ (a star, a quasar, etc.) in a plane (a source
plane) parallel to the lens plane and perpendicular to the line of
sight from the observer, so that the lens plane is between the
observer and the light source. Due to deflection of light by masses
multiple images $S_1,S_2,\dotsc$ of the source may form (cf.\  Fig.\
1). Fig. 2 and Fig. 3 illustrate some further aspects of the lensing
phenomenon.

\begin{figure}[H]
\begin{center}
\includegraphics*[scale=.5]{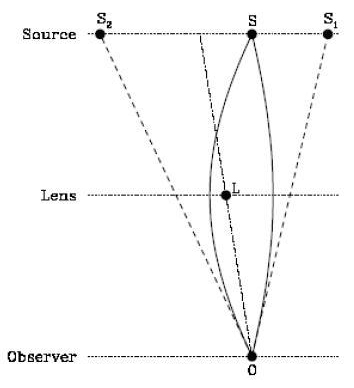}
\end{center}
\caption{The lens $L$ located between source $S$ and observer $O$
produces two images $S_1, S_2$ of the source $S$.}
\end{figure}

\begin{figure}[H]
\begin{center}
\includegraphics*[scale=.5]{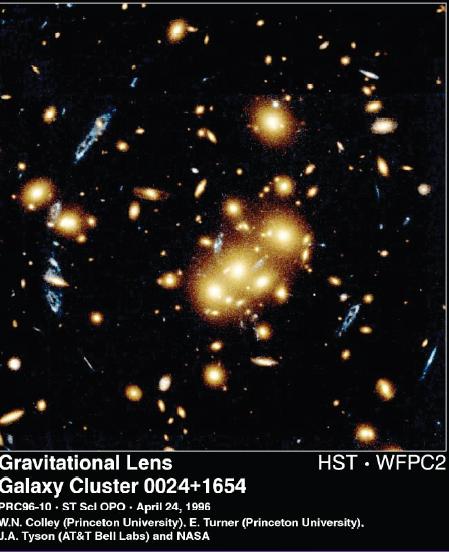}
\end{center}
\caption{Lensing of a galaxy by a cluster of galaxies; the blue
spots are all images of a single galaxy located behind the huge
cluster of galaxies. (Credit: NASA, W. N. Cooley (Princeton), E.
Turrer (Princeton) and J. A. Tyson (AT\&T and Bell Labs).)}
\end{figure}

\begin{figure}[H]
\begin{center}
\includegraphics*[scale=.5]{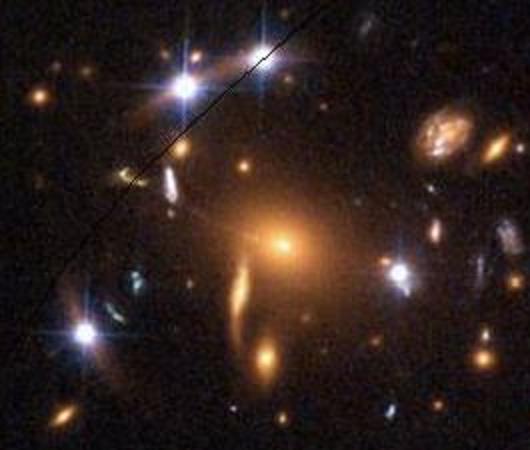}
\end{center}
\caption{The bluish bright spots are the lensed images of a quasar
(i.e., a quasi-stellar object) behind a bright galaxy in the center.
There are 5 images (4 bright + 1 dim), but one cannot really see the
dim image in this figure. (Credit: ESA, NASA, K. Sharon (Tel Aviv
University) and E. Ofek (Caltech).)}
\end{figure}

\section{Lens equation}

In this section we are still assuming that our lens consists of $n$
point masses. Suppose that the light source is located in the
position $w$ (a complex number) in the source plane. Then, the
lensed image is located at $z$ in the source plane while the masses
of the lens $L$ are located at the positions $z_j$, $j=1,\dotsc,n$
in the lens plane. The following simple equation, obtained by
combining Fermat's Principle of Geometric Optics together with basic
equations of General Relativity, connects then the positions of the
lensed images, the source and the positions of the masses which
cause the lensing effect
\begin{equation}
\label{eq2.1}
w=z-\sum_1^n\frac{\sigma_j}{\bar{z}-\bar{z}_j},
\end{equation}
where $\sigma_j$ are some physical (real) constants. For more
details on the derivation and history of the lensing equation
\eqref{eq2.1} we refer the reader to \cite{St}, \cite{NB},
\cite{PLW}, \cite{Wa}. Sometimes, to include the effect caused by an
extra (``tidal'') gravitational pull by an object (such as a galaxy)
far away from the lens masses, the right-hand side of \eqref{eq2.1}
includes an extra linear term $\gamma\bar{z}$, thus becoming
\begin{equation}
\label{eq2.2}
w=z-\sum_1^n\frac{\sigma_j}{\bar{z}-\bar{z}_j}-\gamma\bar{z},
\end{equation}
where $\gamma$ is a real constant. The right-hand side of
\eqref{eq2.1} or \eqref{eq2.2} is called the lensing map. The number
of solutions $z$ of \eqref{eq2.1} (or \eqref{eq2.2}) is precisely
the number of images of the source $w$ generated by the lens $L$.
Letting $r(z)=\sum\limits_1^n\frac{\sigma_j}{z-z_j}+\gamma
z+\bar{w}$, the lens equations \eqref{eq2.1} and \eqref{eq2.2}
become
\begin{equation}
\label{eq2.3}
z-\overline{r(z)}=0,
\end{equation}
where $r(z)$ is a rational function with poles at $z_j$, $j=1,\dotsc,n$ and infinity if $\gamma\ne 0$.

\section{Historical remarks}

The first calculations of the deflection angle by a point mass lens,
based on Newton's corpuscular theory of light and the Law of
Gravity, go back to H. Cavendish and J. Michel (1784), and P.
Laplace (1796)--- cf.\ \cite{Tu}. J. Soldner (1804) --- cf.\
\cite{Wa} is usually credited with the first published calculations
of the deflection angle and, accordingly, with that of the lensing
effect. Since Soldner's calculations were based on Newtonian
mechanics they were off by a factor of $2$. A. Einstein is usually
given credit for calculating the lensing effect in the case of $n=1$
(one mass lens) around 1933. Yet, some evidence has surfaced
recently that he did some of these calculations earlier, around 1912
--- cf.\  \cite{Sa} and references therein. The recent outburst of
activity in the area of lensing is often attributed to dramatic
improvements of optics technology that make it possible to check
many calculations and predictions by direct visualization.

H. Witt \cite{Wit} showed by a direct calculation that for $n>1$ the
maximum number of observed images is $\le n^2+1$. Note that this
estimate can also be derived from the well known Bezout theorem in
algebraic geometry (cf. \cite{KN1,KN2,GH,Wil}). In \cite{MPW} S. Mao
and A. O. Petters and  H. J. Witt showed that the maximum possible
number of images produced by an $n$-lens is at least $3n+1$. A. O.
Petters in \cite{Pe}, using Morse's theory, obtained a number of
estimates for the number of images produced by a non-planar lens. S.
H. Rhie \cite{Rh1} conjectured that the upper-bound for the number
of lensed images for an $n$-lens is $5n-5$. Moreover, she showed in
\cite{Rh2} that this bound is attained for every $n>1$ and, hence,
is sharp. Rhie's conjecture was proved in full in \cite{KN1}.
Namely, we have the following result.

\begin{theorem}\label{thm3.1}
The number of lensed images by an $n$-mass, $n>1$, planar lens cannot exceed $5n-5$ and this bound is sharp \cite{Rh2}. Moreover, the number of images is an even number when $n$ is odd and odd where $n$ is even.
\end{theorem}

The proof of the above result rests on some simple ideas from complex dynamics (cf. \cite{KS,KN2}).

\section{``Thin'' lenses with continuous mass distributions}

If we to replace point masses by a general, real-valued mass
distribution $\mu$, a compactly supported Borel measure in the lens
plane, the lens equation with shear \eqref{eq2.2} becomes
\begin{equation}
\label{eq4.1}
w=z-\int_\Omega\frac{d\mu(\zeta)}{z-\zeta}-\gamma\bar{z}.
\end{equation}
Here $\Omega$ is a bounded domain containing the support of $\mu$.
The case of the atomic measure
$\mu=\sum\limits_1^n\sigma_j\delta_{z_j}$, $\sigma_i\in\mbb{R}$ is
covered by Theorem \ref{thm3.1}. Also, as noted in \cite{KN1}, if we
replace $n$-point-masses by $n$ non-overlapping radially symmetric
masses, the total number of images outside of the region occupied by
$n$-masses is still $5n-5$ when $\gamma=0$, and $\le5n$ when
$\gamma\ne 0$. The reason for that, of course, is that the Cauchy
integral
$$
\int_{\left|\zeta-z_j\right|<R}\frac{d\mu(\zeta)}{z-\zeta},
\quad\left|z-z_j\right|>R
$$
for any radially symmetric measure $\mu=\mu\(\left|\zeta-z_j\right|\)$ is immediately calculated to be equal $\frac c{z-z_j}$, where $c$ is the total mass $\mu$ of the disk $\left\{\zeta:\left|\zeta-z_j\right|<R\right\}$, hence reducing this new situation to the one treated in Theorem \ref{thm3.1}.

Here is another situation that can be treated with help from Theorem \ref{thm3.1}.

Recall that a simply-connected domain $\Omega$ is called a
\textit{quadrature domain} (of order $n$) if $\Omega$ is obtained
from the unit disk $\mbb{D}:=\{z:|z|<1\}$ via a conformal mapping
$\varphi$ that is a rational function of degree $n$,
$\Omega=\varphi(\mbb{D})$. Of course, all poles $\beta_j$,
$j=1,\dotsc,n$ of $\varphi$ will lie outside $\mbb{D}$. Then if,
say, $\mu$ is a uniform mass distribution in $\Omega$, i.e., $\mu =
\mbox{const} \, dx \, dy$, the Cauchy potential term in
\eqref{eq4.1} for $z\notin\overline{\Omega}$ becomes
\begin{equation}
\label{eq4.2} \sum_{j=1}^n\frac{c_j}{z-z_j}, \quad
z_j=\varphi\(\frac1{\overline{\beta_j}}\),
\end{equation}
where the coefficients $c_j$ are determined by the quadrature formula associated with $\Omega$ (cf. \cite{Sh} for details).

Hence, substituting \eqref{eq4.2} into \eqref{eq4.1} we again obtain that for such thin lens $\Omega$ with a uniform density distribution, the number of ``bright'' images outside $\Omega$ cannot exceed $5n-5$ when no shear is present, or $5n$ otherwise.

In this general context the only previously known (to the best of our knowledge) result is the celebrated Burke's theorem \cite{Bu}

\begin{theorem}\label{thm4.1}
A \tu{(}finite\tu{)} number of images produced by a \textbf{smooth}
mass distribution $\mu$ is always odd, provided that $\gamma=0$
\tu{(}no shear\tu{)}.
\end{theorem}

An elegant complex-analytic proof of Burke's theorem can be found in \cite{St}. The crux of the argument is this. Take $w=0$ and let $n_+$, $n_-$ denote, respectively, the number of sense-preserving and sense-reversing zeros of the lens map in \eqref{eq4.1} ($\gamma=0$).

The argument principle applies to harmonic complex-valued functions
in the same way it does to analytic functions. Since the right-hand
side of \eqref{eq4.1} behaves like $O(z)$ near $\infty$, the
argument principle then yields that $1=n_+-n_-$. Thus, giving us the
total number of zeros $N:=n_++n_-=2n_-+1$, an odd number.

\section{Ellipsoidal lens}

Suppose the lens $\Omega:=\left\{\frac{x^2}{a^2}+\frac{y^2}{b^2}\le
1,a>b>0\right\}$ is an ellipse. First assume the mass density to be
constant, say $1$. Let $c:  c^2=a^2-b^2$ be the focal distance of
$\Omega$. The lens equation \eqref{eq4.1} can be rewritten as
\begin{equation}
\label{eq5.1}
\bar{z}-\frac1\pi\int_\Omega\frac{dA(\zeta)}{z-\zeta}-\gamma z=\bar{w},
\end{equation}
where $dA$ denotes the area measure. Using complex Green's formula
(cf. e.g., \cite{St}), we can rewrite \eqref{eq5.1} for
$z\in\mbb{C}\setminus\overline{\Omega}$ as follows:
\begin{equation}
\label{eq5.2} \bar{z}-\frac1{2\pi i}\int_{\partial\Omega}
\frac{\bar{\zeta}\,d\zeta}{z-\zeta}-\gamma z=\bar{w}.
\end{equation}
As is well-known \cite{Sh}, the (analytic) Schwarz function
$S(\zeta)$ for the ellipse defined by $S(\zeta)=\bar{\zeta}$ on
$\partial\Omega$ can be easily calculated and equals
\begin{equation}
\label{eq5.3}
\begin{gathered}
S(\zeta)=\frac{a^2+b^2}{c^2}\,\zeta
-\frac{2ab}{c^2}\(\zeta-\sqrt{\zeta^2-c^2}\,\) \\
=\frac{a^2+b^2-2ab}{c^2}\,\zeta
+\frac{2ab}{c^2}\(\zeta-\sqrt{\zeta^2-c^2}\,\) \\
=S_1(\zeta)+S_2(\zeta).
\end{gathered}
\end{equation}
Note that $S_1$ is analytic in $\overline{\Omega}$, while $S_2$ is
analytic outside $\Omega$ and $S_2(\infty)=0$. This is, of course,
nothing else but the Plemelj--Sokhotsky decomposition of the Schwarz
function $S(\zeta)$ of $\partial\Omega$. From \eqref{eq5.3} and
Cauchy's theorem we easily deduce that for
$z\in\mbb{C}\setminus\overline{\Omega}$ the lens equation
\eqref{eq5.2} reduces to
\begin{equation}
\label{eq5.4}
\bar{z}+\frac{2ab}{c^2}\(z-\sqrt{z^2-c^2}\,\)-\gamma z=\bar{w}.
\end{equation}
Squaring and simplifying, we arrive from \eqref{eq5.4} at a complex quadratic equation
$$
\[\bar{z}+\(\frac{2ab}{c^2}\,z-\gamma\)z\bar{w}\]^2
=\frac{2a^2b^2}{c^2}\(z^2-c^2\)
$$
which is equivalent to a system of two irreducible real quadratic
equations. Bezout's theorem (cf. \cite{KN1,KN2}, \cite{KS},
\cite{GH}) then implies that \eqref{eq5.1} may only have $4$
solutions $z\notin\Omega$. For $z\in\Omega$, using Green's formula
and \eqref{eq5.3} we can rewrite the area integral in \eqref{eq5.1}
\begin{equation}
\label{eq5.5}
\begin{gathered}
-\frac1\pi\int_\Omega\frac{dA(\zeta)}{z-\zeta} =-\bar{z}+\frac1{2\pi
i}\int_{\partial\Omega}
\frac{\bar{\zeta}\,d\zeta}{\zeta-z} \\
=-\bar{z}+\frac1{2\pi i}\int_{\partial\Omega}
\frac{\[S_1(\zeta)+S_2(\zeta)\]}{\zeta-z} \\
=-\bar{z}+S_1(z)=-\bar{z}
+\frac{a^2+b^2-2ab}{c^2}\,z
\end{gathered}
\end{equation}
We have used here that the Cauchy transform of $S_2\mid_{\partial\Omega}$ vanishes in $\Omega$ since $S_2$ is analytic in $\overline{\mbb{C}\setminus\Omega}$ and vanishes at infinity. Substituting \eqref{eq5.5} into \eqref{eq5.1}, we arrive at a linear equation
\begin{equation}
\label{eq5.6}
\(\frac{a^2+b^2-2ab}{c^2}-\gamma\)z=\bar{w}
\end{equation}
for $z\in\Omega$. Equation \eqref{eq5.6}, of course, may only have one root in $\Omega$. Thus, we have proved the following

\begin{theorem}\label{thm5.1}
An elliptic lens $\Omega$ \tu{(}say, a galaxy\tu{)} with a uniform mass density may produce at most four ``bright'' lensing images of a point light source outside $\Omega$ and one \tu{(}``dim''\tu{)} image inside $\Omega$, i.e., at most $5$ lensing images altogether.
\end{theorem}

This type of result has actually been observed experimentally - cf.
Fig. 4, where four bright images are clearly present. It is
conceivable that the dim image is also there but we can't see it
because it is perhaps too faint compared with the galaxy. Of course,
one has to accept Fig. 4 with a grain of salt since we do not expect
``real'' galaxies to have uniform densities. A model of an
elliptical lens, with shear, that produces five images ($4 \,
\mbox{bright} + 1 \, \mbox{dim}$) is given in Fig. 5.

\begin{figure}[H]
\begin{center}
\includegraphics*[scale=.5]{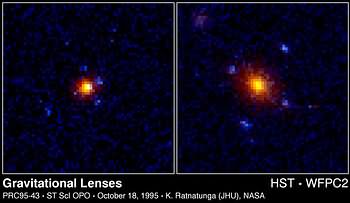}
\end{center}
\caption{Four images of a light source behind the elliptical galaxy.
(Credit: NASA, Kavan Ratnatunga (Johns Hopkins University).)}
\end{figure}

\begin{figure}[H]
\begin{center}
\includegraphics*[scale=.5]{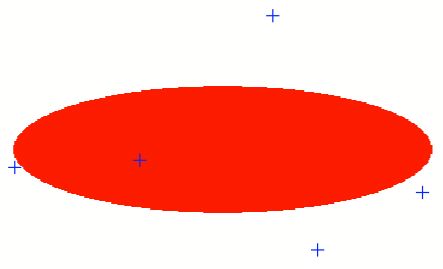}
\end{center}
\caption{A model with five images of a source behind an elliptical
lens with axis ratio $0.5$ and uniform density $2$.}
\end{figure}

We can extend the previous theorem for a larger class of mass
densities. Denote by $q(x,y):=\frac{x^2}{a^2}+\frac{y^2}{b^2}-1$ the
equation of $\Gamma:=\partial\Omega$. Let
$q_\lambda(x,y)=\frac{x^2}{a^2+\lambda}+\frac{y^2}{b^2+\lambda}-1$,
$-b^2<\lambda<0$ stand for the equation of the boundary
$\Gamma_\lambda:=\partial\Omega_\lambda$ of the ellipse
$\Omega_\lambda$ confocal with $\Omega$.

The celebrated MacLaurin's theorem (cf. \cite{Kh}) yields that for
any $z\in\mbb{C}\setminus\overline{\Omega}$
\begin{equation}
\label{eq5.7}
\frac1{\Area\(\Omega_\lambda\)}\int_{\Omega_\lambda}\frac{dA(\zeta)}{\zeta-z}
=\frac1{\Area(\Omega)}\int_\Omega\frac{dA(\zeta)}{\zeta-z}.
\end{equation}
Thus, if we denote by $u(z,\lambda)$ the Cauchy potential of $\Omega_\lambda$ evaluated at $z\in\mbb{C}\mid\overline{\Omega}$ we obtain from \eqref{eq5.7}
\begin{equation}
\label{eq5.8} u\(z, \lambda\)=c(\lambda) u_\Omega(z,0),
\end{equation}
where
\begin{equation}
\label{eq5.9}
c(\lambda)=\frac{\Area\(\Omega_\lambda\)}{\Area(\Omega)}
=\frac{\(a^2+\lambda\)^{1/2}\(b^2+\lambda\)^{1/2}}{ab}.
\end{equation}
Hence,
\begin{equation}
\label{eq5.10}
\frac{\partial u_\lambda(z,\lambda)}{\partial\lambda}
=c'(\lambda) u_\Omega(z).
\end{equation}
So, if the mass density $\mu(\lambda)$ in $\Omega$ only depends on
the elliptic coordinate $\lambda$, i. e., is constant on ellipses
confocal with $\Omega$ inside $\Omega$, its potential outside
$\Omega$ equals
\begin{equation}
\label{eq5.11} u_{\mu,\Omega}(z)=c u_\Omega(z).
\end{equation}
The constant $c$ is easily calculated from
\eqref{eq5.9}--\eqref{eq5.10} and equals
\begin{equation}
\label{eq5.12} c=\int_{-b^2}^0\mu(\lambda)c'(\lambda)\,d\lambda.
\end{equation}
It is, of course, natural for physical reasons to assume that
$\mu(\lambda)\uparrow\infty$ at the ``core'' of $\Omega$ (i.e., when
$\lambda\downarrow-b^2$), the focal segment $[-c,c]$. Yet, from
\eqref{eq5.12} since \eqref{eq5.9} yields
$c'(\lambda)=O\(\(b^2+\lambda\)^{-1/2}\)$ near $\lambda_0=-b^2$, it
follows that $\mu(\lambda)$ should not diverge at the core faster
than say $O\(\(b^2+\lambda\)^{-1/2+\epsilon}\)$ for some positive
$\epsilon$, so the integral \eqref{eq5.12} converges. Substituting
\eqref{eq5.11} into the lens equation \eqref{eq5.1} with constant
density replaced by the density $\mu(\lambda)$ and following again
the steps in \eqref{eq5.2}--\eqref{eq5.4} we arrive at the following
corollary.

\begin{corollary}\label{cor5.1}
An elliptic lens $\Omega$ with mass density that is constant inside
$\Omega$ on the ellipses confocal with $\Omega$  may produce at most
four ``bright'' lensing images of a point light source outside
$\Omega$.
\end{corollary}

\section{Einstein rings}

For a one-point mass at $z_1$ lens with the source at $w=0$ the lens
equation \eqref{eq2.1} without shear becomes
\begin{equation}
\label{eq6.1}
z-\frac c{\bar{z}-\bar{z}_1}=0.
\end{equation}
As was already noted by Einstein (cf. \cite{St,Wa,NB} and references
cited therein), \eqref{eq6.1} may have two solutions (images) when
$z_1\ne 0$ and a whole circle (``Einstein ring'') of solutions when
$z_1=0$, in other words when the light source, the lens and the
observer coalesce - cf. \, Fig. 6 and Fig. 7.

\begin{figure}[H]
\begin{center}
\includegraphics*[scale=.5]{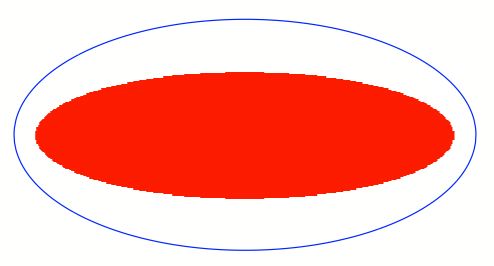}
\end{center}
\caption{A model of an elliptical Einstein ring surrounding an
elliptical lens with axis ratio $0.5$ and uniform density $2$.
The shear in this case must be specially chosen to produce the ring
instead of point images. Note that the ring is an ellipse confocal
with the lens - cf. Thm. 6.1 .}
\end{figure}

\begin{figure}[H]
\begin{center}
\includegraphics*[scale=.1]{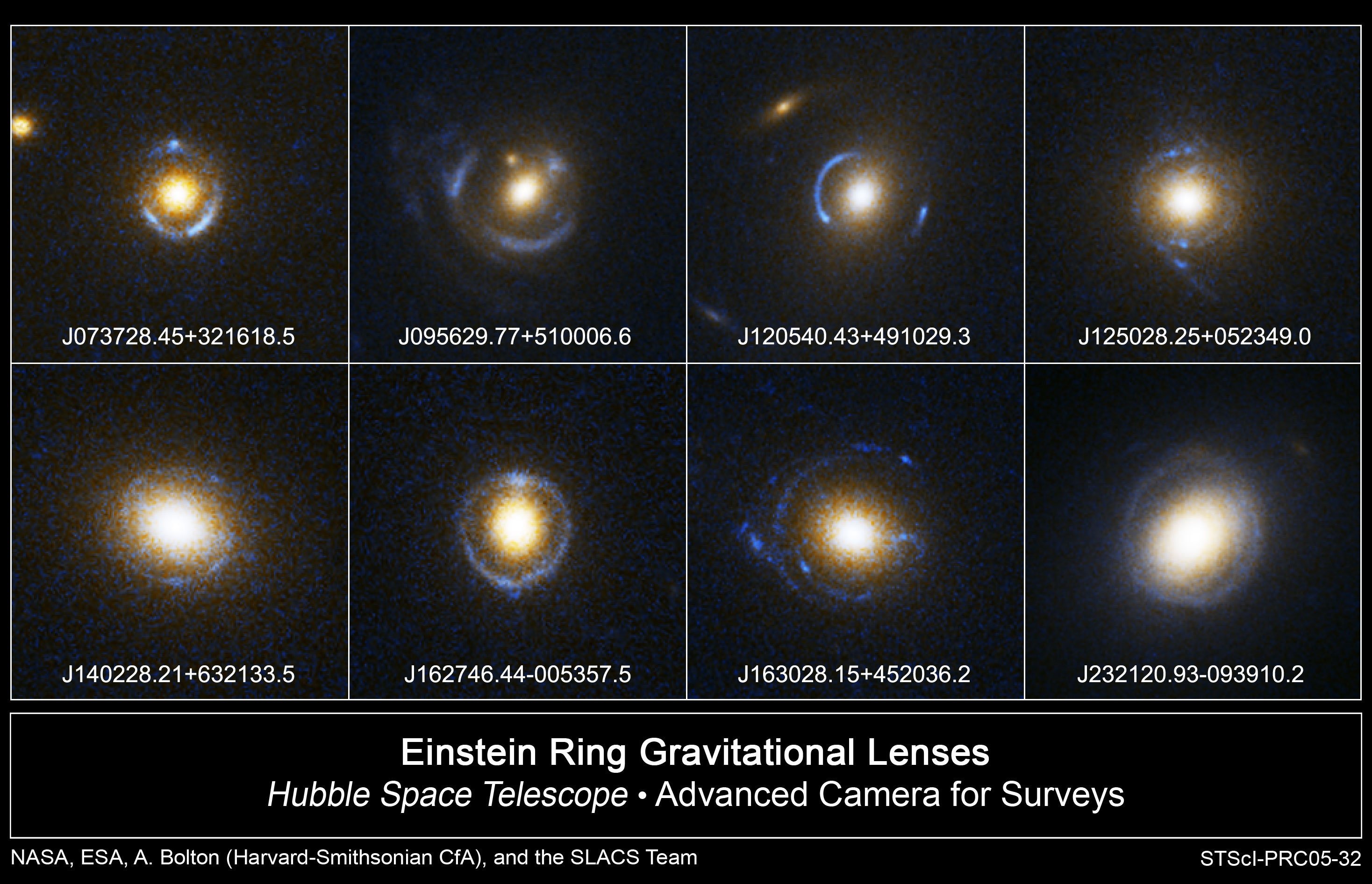}
\end{center}
\caption{Einstein rings. The sources in these observed ``realistic''
lenses are actually extended, and that is why we see sometimes arcs
rather than whole rings. (Credit: ESA, NASA and the SLACS survey
team: A. Bolton (Harvard / Smithsonian), S. Burles (MIT), L.
Koopmans (Kapteyn), T. Treu (UCSB), and L. Moustakas (JPL/Caltech).)
}
\end{figure}

As the following simple theorem shows the ``ideal'' Einstein rings
are limited to ellipses and circles  in much more general
circumstances.

\begin{theorem}\label{thm6.1}
Let $\Omega$ be any planar \tu{(}``thin''\tu{)} lens with mass
distribution $\mu_e$. If lensing of a point source produces a
bounded ``image'' curve outside of the lens $\Omega$, it must either
be a circle when the external shear $\gamma=0$ or an ellipse.
\end{theorem}

\begin{proof}
First consider a simpler case when $\gamma=0$. If the lens produces
an image curve $\Gamma$ outside $\Omega$, the lens equation
\eqref{eq4.1} becomes
\begin{equation}
\label{eq6.2}
\bar{z}-\bar{w}=\int_\Omega\frac{d\mu(\zeta)}{z-\zeta},
\end{equation}
for all $z\in\Gamma$. Note that $\Gamma$ being bounded and also
being a level curve of a harmonic function must contain a closed
loop surrounding $\Omega$ \cite{Wil}. Without loss of generality, we
still denote that loop by $\Gamma$. The right-hand side $f(z)$ of
\eqref{eq6.2} is a bounded analytic function in the unbounded
complement component $\tilde{\Omega}_\infty$ of $\Gamma$ that
vanishes at infinity. Hence $(z-w)f(z)$ is still a bounded and
analytic function in $\tilde{\Omega}_\infty$ equal to $|z-w|^2>0$ on
$\Gamma:=\partial\tilde{\Omega}_\infty$. Hence $(z-w)f(z)=\const$
and $\Gamma$ must be a circle centered at $w$.

Now suppose $\gamma\ne 0$. Once again we shall still denote by
$\Gamma$  a closed Jordan loop surrounding $\Omega$. Denote by
$\tilde{\Omega}$ the interior of $\Gamma$,
$\tilde{\Omega}_\infty=\mbb{C}\setminus\clos\(\tilde{\Omega}\)$.
Also, by translating we can assume that the position of the source
$w$ is at the origin.

The equation \eqref{eq4.1} now reads
\begin{equation}
\label{eq6.3}
\bar{z}=\int_\Omega\frac{d\mu(\zeta)}{z-\zeta}+\gamma z,
\quad z\in\Gamma.
\end{equation}
In other words the right-hand side of \eqref{eq6.3} represents the Schwarz function $S$ of $\Gamma$, analytic in $\mbb{C}\setminus\supp\mu$ with a simple pole at $\infty$. It is well-known that this already implies that $\Gamma$ must be an ellipse (cf. \cite{Sh,GS}) and references therein. For the reader's convenience we supply a simple proof.

Applying Green's formula to \eqref{eq6.3} yields (cf.\  \eqref{eq5.1}--\eqref{eq5.2}) that for all $z\in\tilde{\Omega}_\infty$
\begin{equation}
\label{eq6.4} \int_\Omega\frac{d\mu(\zeta)}{z-\zeta}
=\frac1\pi\int_{\tilde\Omega}\frac{dA(\zeta)}{z-\zeta}, \quad
z\in\tilde{\Omega}:=\mbb{C}\setminus\tilde{\Omega}_\infty.
\end{equation}
Let
\begin{equation}
\label{eq6.5}
h(z):=\frac1\pi\int_{\tilde\Omega}\frac{dA(\zeta)}{z-\zeta}-\bar{z},\quad
z\in\tilde{\Omega}.
\end{equation}
Then, $h(z)$ is analytic in $\tilde\Omega$ (cf. \eqref{eq5.2}) and, in view of \eqref{eq6.3} and \eqref{eq6.4}
\begin{equation}
\label{eq6.6}
h\mid_\Gamma=\int_\Omega
\left.\frac{d\mu(\zeta)}{z-\zeta}\right|_\Gamma
-\left.\bar{z}\right|_\Gamma=\gamma z\mid_\Gamma.
\end{equation}
Thus, $h(z)$ is a linear function and since \eqref{eq6.5} for $z\in\tilde{\Omega}$
\begin{equation}
\label{eq6.7}
h(z):=\frac{1}{2}\grad
\[\frac1\pi\int_{\tilde\Omega}\log|z-\zeta|\,dA(\zeta)-|z|^2\]
\end{equation}
we conclude from \eqref{eq6.7} that the potential of $\tilde\Omega$
$$
u_{\tilde\Omega}(z)=\frac1{2\pi}\int_{\tilde\Omega}\log|z-\zeta|\,dA(\zeta),
\quad z\in\tilde{\Omega}
$$
equals to a quadratic polynomial inside $\tilde{\Omega}$. The
converse of the celebrated theorem of Newton due to P. Dive and N.
Nikliborc (cf. \cite[Ch. 13--14]{Kh} and references therein) now
yields that $\tilde{\Omega}$ must be an interior of an ellipse,
hence $\Gamma:=\partial\tilde{\Omega}=\partial\tilde{\Omega}_\infty$
is an ellipse.
\end{proof}

\begin{remark}\label{rem6.1}
One immediately observes that since the converse to Newton's theorem holds in all dimensions the last theorem at once extends to higher dimensions if one replaces the words ``image curve'' by ``image surface''.
\end{remark}

\section{Final remarks}

\begin{enumerate}
\item  The densities considered in \S5 are less important from the physical viewpoint than so-called
``isothermal density'' which is obtained by projecting onto the lens
plane the ``realistic'' three-dimensional density $\sim 1/\rho^2$,
where $\rho$ is the (three-dimensional ) distance from the origin.
This two-dimensional density could be included into the whole class
of densities that are constant on all ellipses \textit{homothetic}
rather than confocal with the given one. The reason for the term
``isothermal'' is that when a three-dimensional galaxy has density
$\sim 1/\rho^2$ the gas in the galaxy has constant temperature (cf.
\cite{KMW} and the references therein).

Recall that the Cauchy potential of the ellipse $\Omega:=\left\{\frac{x^2}{a^2}+\frac{y^2}{b^2}\le 1,a>b>0\right\}$
outside of $\Omega$ (cf. \eqref{eq5.2}--\eqref{eq5.4}) equals
\begin{equation}
\label{eq7.1}
u_0(z):=k\(z-\sqrt{z^2-c^2}\,\),
\quad z\in\mbb{C}\setminus\overline{\Omega},
\quad c^2=a^2-b^2,
\end{equation}
where $k=2ab/c^2$ is a constant. Replacing the ellipse $\Omega$ by a
homothetic ellipse
$\Omega_t:=t\Omega:=\frac{x^2}{a^2}+\frac{y^2}{b^2}\le t^2$,
\,$0<t<1$. We obtain using \eqref{eq7.1} for
$z\notin\overline{\Omega}$:
\begin{equation}
\label{eq7.2}
\begin{split}
u(z,t): &=\int_{t\Omega}\frac{dA(\zeta)}{\zeta-z}
=t^2\int_\Omega\frac{dA(\zeta)}{t\zeta-z} \\
&=tu\(\frac zt;1\)=k\(z-\sqrt{z^2-c^2t^2}\,\).
\end{split}
\end{equation}
Thus,
\begin{equation}
\label{eq7.3}
\frac{\partial u(z,t)}{\partial t}=k\,\frac{c^2t}{\sqrt{z^2-c^2t^2}}.
\end{equation}

So, if the ``isothermal'' density $\mu=\frac1t$ on
$\partial\Omega_t$ inside $\Omega$ (ignoring constants), we get from
\eqref{eq7.3} that the Cauchy potential of such mass distribution
outside $\Omega$ equals
\begin{equation}
\label{eq7.4} u_\mu(z):=C_0\int_0^1\frac{dt}{\sqrt{z^2-c^2t^2}},
\end{equation}
where the constant $C_0$ depends on $\Omega$ only. This is a
transcendental function (one of the branches of $\arcsin\frac cz$),
dramatically different from the algebraic potential in
\eqref{eq7.1}. The lens equation \eqref{eq4.1}now becomes
\begin{equation}
\label{eq7.5}
z-C_0\int_0^1\frac{dt}{\sqrt{\bar{z}^2-c^2t^2}}-\gamma\bar{z}=w.
\end{equation}

To the best of our knowledge the precise bound on the maximal
possible number of solutions (images) of \eqref{eq7.5} is not known.
Up to today, no more than $5$ images ($4$ bright $+1$ dim) have been
observed. However, in \cite{KMW} there have been constructed
explicit models depending on parameters $a$, $b$ and $0<\gamma<1$
having $9$ (i. e., $8+1$) images. The equation \eqref{eq7.5}
essentially differs from all the lens equations considered in this
paper since it involves estimating the number of zeros of a
transcendental harmonic function with a simple pole at $\infty$. At
this point, we are even reluctant to make a conjecture regarding
what this maximal number might be.

Note, that in case of a circle $\Omega=\left\{x^2+y^2<1\right\}$
with any radial density $\mu:=\varphi(r)$, $r=\sqrt{x^2+y^2}<1$, the
situation is very simple. The Cauchy potential $u(z)$ outside
$\Omega$, as was noted earlier, equals
\begin{equation}
\label{eq7.6}
\frac cz,\quad |z|>1,
\end{equation}
where $c$ is a constant. Hence, outside $\Omega$ the lens equation becomes
\begin{equation}
\label{eq7.7}
z-\frac c{\bar{z}}-\gamma\bar{z}=w,
\end{equation}
a well-known Chang--Refsdal lens (cf., e.g., \cite{AE}) that may
have at most $4$ solutions except for the degenerate case
$\gamma=w=0$, when the Einstein ring appears. In particular, when
$\gamma=0,\, w\neq 0$, such mass distribution may only produce two
bright images outside $\Omega$. For $z:|z|<1$ inside the lens the
potential is still calculated by switching to polar coordinates:
\begin{equation}
\label{eq7.8}
\begin{gathered}
u(z):=\int_0^1\int_0^{2\pi}\frac{\varphi(r)rdrd\theta}{re^{i\theta}-z} \\
=\int_{|z|}^1\varphi(r)dr\int_0^{2\pi}\frac{rd\theta}{re^{i\theta}-z}
+\int_0^{|z|}\varphi(r)rdr\int_0^{2\pi}\frac{d\theta}{re^{i\theta}-z} \\
=\int_{|z|}^1\varphi(r)dr\int_0^{2\pi}\(\sum_0^\infty\(\frac zr\)e^{-i(n+1)\theta}\)d\theta \\
+\frac1z\int_0^{|z|}\varphi(r)rdr\int_0^{2\pi}\(\sum_0^\infty\(\frac{re^{i\theta}}z\)^n\right. d\theta \\
=\frac{2\pi}z\int_0^{|z|}\varphi(r)rdr.
\end{gathered}
\end{equation}
In particular, for the ``isothermal'' density $\varphi(r)\sim\frac1r$, \eqref{eq7.8} yields for $z:|z|<1$
$$
u(z)=\frac{2\pi}z\,|z|,
$$
so the lens equation \eqref{eq7.7} becomes
\begin{equation}
\label{eq7.9}
\bar{z}-\frac cz\,|z|-\gamma z=\bar{w},
\end{equation}
where $c$ is a real constant. Equation \eqref{eq7.9} can have at
most two solutions \textit{inside} $\Omega$ (only one, if
$\gamma=0$), again, excluding the degenerate case of the Einstein
ring. Furthermore, since Burke's theorem allows only an odd number
of images, the total maximal number of images for an isothermal
sphere cannot exceed $5$ ($4\, \mbox{bright} + 1 \, \mbox{dim}$) as
before (or $\le 3, \,\mbox{i. e.,}(2+1)$ if $\gamma=0$). Note, that
strictly speaking, Burke's theorem cannot be applied to the
isothermal density because of the singularity at the origin. Yet,
since the density is radial and smooth everywhere excluding the
origin and because it is clear from \eqref{eq7.9} that the origin
cannot be a solution, Burke's theorem does apply yielding the above
conclusion.

\item  The problem of estimating the maximal number of ``dim'' images inside the lens formed by a uniform mass-distribution inside a
quadrature domain $\Omega$ (cf. \S4) of order $n$ is challenging. In
this case the Cauchy potential in \eqref{eq4.1} inside $\Omega$
equals to the ``analytic'' part of the Schwarz function $S(z)$. It
is known that $S(z)$ is an algebraic function of degree at most
$2n$. Yet, the sharp bounds, similar to those in Theorem
\ref{thm3.1}, for the number of zeros of harmonic functions of the
form $\bar{z}-a(z)$, where $a(z)$ is an algebraic function, are not
known.

\item  Another interesting and difficult problem would be to study the maximal number of images by a lens
consisting of several elliptical mass distributions. Some rough
estimates based on Bezout's theorem can be made by imitating the
calculations in \S5. Yet, even for $2$ uniformly distributed masses
these calculations give a rather large possible number of images
($\le 15$) while, so far, only $5$ images by a two galaxies lens and
$6$ images by a three galaxies lens have been observed - cf.
\cite{KW,Win}.
\end{enumerate}

\bibliographystyle{plain}
\bibliography{g-lensing}

\end{document}